\documentclass[journal]{IEEEtran}

\ifCLASSINFOpdf
\else
   \usepackage[dvips]{graphicx}
\fi

\usepackage{cite}
\usepackage{amsmath,amssymb,amsfonts,amsthm}
\usepackage{graphicx}
\usepackage{textcomp}

\usepackage{xcolor}
\usepackage[colorlinks=true,  linkcolor=black, citecolor=blue, breaklinks=true, urlcolor=blue]{hyperref}

\usepackage{graphicx}

\usepackage[ruled,vlined]{algorithm2e}
\usepackage{xcolor}
\usepackage{booktabs}
\usepackage{url}
\usepackage[latin1]{inputenc}

\definecolor{darkgreen}{rgb}{0.2, 0.6, 0.05}
\usepackage{enumitem}

\usepackage{subcaption}

\newcounter{clemma}
\newtheorem{lemma}[clemma]{Lemma}

\usepackage{pgfplots}
\usepackage{tikz}
\usepgfplotslibrary{groupplots}
\pgfplotsset{compat=1.13}

\newlength\fwidth

\usepackage{algorithmic}

\usetikzlibrary{arrows,petri,topaths}

\newcommand{\E}{\mathbb{E}}                  % expectation
\newcommand{\bo}[1]{\mathbf{#1}}              % boldface math 
\newcommand{\bom}[1]{\boldsymbol{#1}}    % boldface math (for greek letters)

\newcommand{\beq}{\begin{equation}}
\newcommand{\eeq}{\end{equation}}
\newcommand{\bmat}{\begin{pmatrix}}
\newcommand{\emat}{\end{pmatrix}}

\newcommand{\A}{{\bf A}}
\renewcommand{\S}{{\bo S}} % SAMPLE COVARIANCE MATRIX
\newcommand{\x}{\bo x}

\newcommand{\M}{\bom \Sigma}

\DeclareMathOperator{\tr}{tr}

\DeclareMathOperator{\cov}{cov}

\begin{document}

\title{Affine equivariant Tyler's M-estimator applied to tail parameter learning of elliptical distributions}

\author{
Esa~Ollila,~\IEEEmembership{Senior Member,~IEEE,} Daniel P. Palomar,~\IEEEmembership{Fellow,~IEEE,} and Fr\'ed\'eric~Pascal~\IEEEmembership{Senior Member,~IEEE} 
\thanks{Esa Ollila is with the Department of Information and Communications Engineering, Aalto University, FI-00076 Aalto, Finland (e-mail: esa.ollila@aalto.fi). Daniel P. Palomar is with the Hong Kong University of Science and Technology, Hong Kong (e-mail: palomar@ust.hk). Fr\'ed\'eric Pascal is with the CNRS Universit\'e Paris-Saclay, CentraleSup\'elec, 91190 Gif-sur-Yvette, France (e-mail:frederic.pascal@centralesupelec.fr).}}

\markboth{Affine equivariant Tyler's M-estimator}
{Ollila \MakeLowercase{\textit{et al.}}}
\maketitle

\begin{abstract} We propose estimating the scale parameter (mean of the eigenvalues) of the scatter matrix of an unspecified elliptically symmetric distribution using weights obtained by solving Tyler's M-estimator of the scatter matrix. The proposed Tyler's weights-based estimate (TWE)  of scale is then used to construct an affine equivariant Tyler's M-estimator as a weighted sample covariance matrix using normalized Tyler's weights. We then develop a unified framework for estimating the unknown  tail parameter of the elliptical distribution (such as the degrees of freedom (d.o.f.) $\nu$ of the multivariate $t$ (MVT) distribution). Using the proposed TWE of scale, a new robust estimate of the d.o.f. parameter of MVT distribution is proposed with excellent performance in heavy-tailed scenarios,  outperforming other competing methods. R-package is available that implements the proposed method.  
\end{abstract}

\IEEEpeerreviewmaketitle

\section{Introduction}
\label{sec:intro}

\IEEEPARstart{W}{e model} the observed $p$-variate observations $\x_1,\ldots,\x_n$ as independent and identically distributed (i.i.d.) random samples from an unspecified centered  (i.e., symmetric around the origin) elliptically symmetric (ES) distribution \cite{fang1990symmetric,ollila2012complex}. A continuous random vector $\x \in \mathbb{R}^p$ has centered ES distribution if 
it possesses  a probability density function (pdf) of the form
\begin{equation*}
f(\x)=C_{p,g}    |\M|^{-1/2} g ( \x^\top \M^{-1} \x),
\end{equation*}
where $g:\mathbb{R}_{\geq 0} \to \mathbb{R}_{> 0}$ is called the density
generator,  $\M \succ 0$ is the positive definite symmetric matrix parameter, called the \emph{scatter matrix}, and $C_{p,g}$ is a normalizing constant ensuring that $f(\x)$ integrates to 1. We let $\x \sim \mathcal E_p(\mathbf{0},\M,g)$ to denote this case.  For example, the centered multivariate normal (MVN) distribution  $\mathcal N_p(\mathbf{0},\M)$ is obtained when 
$g(t) = \exp(-t/2)$ while the multivariate $t$ (MVT) distribution with $\nu>0$ degrees of freedom (d.o.f.)  is  obtained when 
\beq \label{eq:pdf_t-dist} 
g(t)=(1+  t/\nu)^{-(p+\nu)/2}.
\eeq 
Parameter $\nu>0$ is a tail parameter of the density. 
For $\nu \to \infty$, the MVT distribution reduces to the MVN distribution, while $\nu=1$ corresponds to the multivariate Cauchy distribution.  Also, many other subclasses of ES distributions are parametrized by a density generator indexed by an additional tail parameter that is unknown in practice and needs to be estimated. Learning this unknown parameter is also one of the goals of this paper. 

We are interested in estimating the scale parameter of the scatter matrix, defined as the mean of its eigenvalues, 
\beq \label{eq:eta}
\eta= \frac{\tr(\M)}{p} = \frac{1}{p} \sum_{i=1}^p \lambda_i , 
\eeq
where $\lambda_i >0$ denotes the $i$th eigenvalue of $\M$. 
 Formally, $\eta \equiv \eta(\M)$ is  a \emph{scale} parameter if it verifies $\eta(\mathbf{I}) = 1$ and  $\eta(a\boldsymbol{\Sigma}) = a \eta(\boldsymbol{\Sigma})$ for all $a>0$ \cite{paindaveine2008canonical}. 

Tyler's M-estimator \cite{tyler1987distribution} is a popular robust M-estimator of the scatter matrix that has been extensively studied both in signal processing and statistics literature (e.g, \cite{sun2014regularized,pascal2014generalized,ollila2014regularized,chen2011robust,couillet2014large,soloveychik2015performance,zhang2016marvcenko,goes2020robust,romanov2023tyler}). Tyler's M-estimator is 
defined as the solution to the fixed-point equation
\beq \label{eq:Mest-tyler}
\hat \M =  \frac{1}{n} \sum_{i=1}^n  \cfrac{p}{ \x_i^\top \hat{\M}^{-1} \x_i}  \,  \x_i \x_i^\top  
\triangleq \mathcal H(\hat \M ; \{\x_i\}). 
\eeq 
Note that me way also write the map $\mathcal H(\cdot; \cdot)$ in the form
$$
\mathcal H(\hat \M ; \{\x_i\}) = \frac{1}{n} \sum_{i=1}^n  \hat w_i   \x_i \x_i^\top  \  \mbox{ with } \
\hat w_i = \cfrac{p}{ \x_i^\top \hat \M^{-1} \x_i} ,
$$
where $\hat w_i$, $i=1,\ldots, n,$ are referred to as \emph{Tyler's weights}.  
Tyler's M-estimator  is  unique only up to a scaling factor, and therefore a common convention  is to consider a solution that verifies  $\tr(\hat{\M})=p$.  Thus, Tyler's M-estimator is actually an estimator of a \emph{shape matrix} (normalized scatter matrix) $\boldsymbol{\Lambda}$, defined by 
$
\boldsymbol{\Lambda} =  \M/\eta = p \M/\tr(\M), 
$
and verifying $\tr(\boldsymbol{\Lambda})=p$.

In this paper, we propose an estimator of the scale $\eta$ based on Tyler's weights $\hat w_i$. The proposed scale estimate along with  Tyler's M-estimator $\hat \M$ are then jointly used for constructing affine equivariant robust  estimates of the scatter matrix $\M$ and the covariance matrix $\mathbf{R}=\cov(\x)$ (or their shrinkage versions). These developments are described in \autoref{sec:TWE}.  Then, in \autoref{sec:tail}, we propose a unified framework allowing to estimate the tail parameter of the elliptical distribution using the proposed scale statistic $\hat \eta$. In the case of the MVT distribution, this leads to a new estimate of the d.o.f. parameter based on Tyler's weights. Finally, Section \ref{sec:simu} demonstrates the relevance of the proposed approach on simulated data, with concluding remarks in Section \ref{sec:conclusion}. In the R package \texttt{fitHeavyTail} \cite{package_fitHeavyTail}, the function \texttt{fit\_Tyler} implements this method.

\section{Estimate of scale, scatter, and covariance matrix based on Tyler's weights}  \label{sec:TWE}

Assuming that $\x \sim   \mathcal E_p(\mathbf{0},\boldsymbol{\Sigma},g)$  has finite 2nd-order moments, then its covariance matrix, $\mathbf{R} =  \E[ \x\x^\top]$ satisfies 
\beq \label{eq:theta} 
\mathbf{R} = \theta  \cdot \boldsymbol{\Sigma} \quad \mbox{for } \quad\theta=  \frac{\E[ r^2] }{p} , 
\eeq 
where  $r^2 = \|\M^{-1/2} \x\|^2$ is the squared Mahalanobis distance of $\x$ w.r.t. $\M$, whose pdf is given by  
\beq \label{eq:rsq_pdf}
f_{r^2}(t) = C t^{p/2 -1} g(t) ,
\eeq
where $C= \int_0^\infty t^{p/2 -1} g(t)  \mathrm{d} t$. Hence pdf of $r^2$ has a one-to-one correspondence with density generator $g$. 
From \eqref{eq:theta}, we notice that the scatter matrix $\boldsymbol{\Sigma}$ is proportional to the covariance matrix $\mathbf{R}$ (assuming $\mathbf{R}$  exists). 
In the MVN case,  $\theta=1$, while for the MVT distribution with density generator as in \eqref{eq:pdf_t-dist} one obtains $\theta= \nu/(\nu-2)$ for all $\nu >2$.

\subsection{Estimate of scale} 

As mentioned earlier, Tyler's M-estimator $\hat \M$ loses information of the scale $\eta$.  However, it is yet possible to construct an estimate of $\eta$   from  Tyler's weights. 
Our Tyler's weights-based estimate (TWE) of scale is defined as  the harmonic mean of reciprocal of weights, $1/\hat w_i$'s, that is,   
\begin{equation} \label{eq:hat_tau3}
\hat \eta_{\text{TWE}} =  \bigg( \frac{1}{n} \sum_{i=1}^n \hat w_i  \bigg)^{-1}   =  \bigg( \frac{p}{n} \sum_{i=1}^n  [\x_{i}^\top \hat{\M}^{-1} \x_{i}]^{-1} \bigg)^{-1} . 
\end{equation} 
One can esily verify that this is a proper scale estimate  in the sense that if $\hat \eta_{\text{TWE}}^*$ is computed on scaled observations, $\x_i^* = c \cdot\x_i$, $i=1,\ldots, n$, then $\hat \eta_{\text{TWE}}^* = c^2  \cdot \hat \eta_{\text{TWE}}$.  This follows because Tyler's M-estimator $\hat \M$ in \eqref{eq:Mest-tyler}  with trace constraint is invariant to scaling the data, so $\hat \M^* = \hat \M$. 

The proposed estimate \eqref{eq:hat_tau3} can also be motivated from the following result  derived in the high-dimensional random matrix theory (RMT) regime, where $p,n \to \infty$ with $n>p$ and  their ratio tending to constant: $
p/n \to c \in (0,1).$
Namely, let $ \hat{\M}$ be Tyler's M-estimator in \eqref{eq:Mest-tyler} verifying $\mathrm{tr}(\hat \M)=p$. Then, it was shown in \cite{zhang2016marvcenko,romanov2023tyler} that $\max_{\ell}  | \eta \hat w_{\ell} - 1 | \to 0$ almost surely.  The authors in \cite{zhang2016marvcenko} derived this result for the case that data is i.i.d. Gaussian $\mathcal N_p(\mathbf{0}, \mathbf{I})$ while \cite{romanov2023tyler} extended these results for more general distributions.  Thus since $1/\hat w_{i}$ concentrate on $\eta$, 
\eqref{eq:hat_tau3} is a natural robust estimator of scale.  Many other robust scale statistics could be constructed from $1/\hat w_i$'s, such as the median, trimmed mean, etc. In the next subsection, we illustrate why the proposed harmonic mean \eqref{eq:hat_tau3} is the most natural. 

\subsection{Affine equivariant estimate of scatter matrix} 

 Using the scale estimate $\hat \eta_{\text{TWE}}$ and Tyler's M-estimator $\hat \M$ (with $\tr(\hat \M)=p$), we can form an estimate of the scatter matrix 
\begin{align} \label{eq:Tyler_Scatter} 
\hat \M_{\text{TWE}}  &= \hat \eta_{\textup{TWE}}   \cdot \hat \M  
\end{align} 
referred  to as \emph{TWE of scatter matrix}. 
Thus $\hat \eta_{\textup{TWE}}$ is scale statistic derived from $\hat \M_{\text{TWE}}$ since $\hat \eta_{\textup{TWE}}= \tr(\hat \M_{\text{TWE}})/p$. Equivalently, the trace of $\hat \M_{\text{TWE}}$ can be easily computed as the harmonic mean of Tyler's quadratic form:  
$$\tr(\hat \M_{\text{TWE}}) = \bigg( \frac{1}{n} \sum_{i=1}^n  [\x_{i}^\top \hat{\M}^{-1} \x_{i}]^{-1} \bigg)^{-1}.$$

Recalling \eqref{eq:hat_tau3} we can write \eqref{eq:Tyler_Scatter} in the following more intuitive form: 
\begin{subequations} \label{eq:TylerScatter}
   \begin{align}
\hat \M_{\text{TWE}}   &=   \frac{1}{n}\sum_{i=1}^n   \hat v_i \x_i \x_i^\top,  \label{eq:TylerScatter2a}  \\ 
\hat v_{i} &=  \frac{\hat w_{i}}{\frac{1}{n}\sum_{\ell=1}^n \hat w_\ell } = \frac{  [\x_{i}^\top \hat \M^{-1} \x_{i}]^{-1} } { \frac{1}{n}\sum_{\ell=1}^n  [\x_\ell^\top \hat \M^{-1} \x_\ell]^{-1} } \label{eq:TylerScatter2b} 
\end{align}
\end{subequations} 
where $\hat v_1,\ldots,\hat v_n$ are normalized Tyler's weights that verify $\frac{1}{n}\sum_{i =1}^n \hat v_i = 1$.  Eq. \eqref{eq:TylerScatter2a} and \eqref{eq:TylerScatter2b}  illustrate that $\hat \M_{\text{TWE}}$ is a \emph{weighted sample covariance matrix (SCM)} with weights $\hat v_{i}$. 

Finally, we draw the parallel of Tyler's M-estimating equation and our estimator \eqref{eq:TylerScatter}. First, note that Tyler's M-estimating equation \eqref{eq:Mest-tyler} verifies broader invariance than just invariance with respect to scaling of the data matrix. Namely, denoting the unit-norm normalized data by $\tilde \x_i = \x_i/\| \x_i \|$, $i=1,\ldots,n$,  one can easily verify that the fixed-point equation in  
\eqref{eq:Mest-tyler} can be rewritten as 
$\hat \M  = \mathcal H(\hat \M ; \{\tilde \x_i\})$, so based on normalized data.  
Furthermore, since $\tr(\hat \M) =p $, one has:
\begin{equation*}
\tr(\hat \M) = \frac{1}{n} \sum_{i=1}^n  \cfrac{p}{ \tilde\x_i^\top \hat{\M}^{-1} \tilde\x_i}  \,  \tr\left(\tilde\x_i \tilde\x_i^\top\right) = p,
\end{equation*}
or equivalently $\frac{1}{n} \sum_{i=1}^n  [\tilde\x_i^\top \hat{\M}^{-1} \tilde\x_i]^{-1} = 1$ since $\tr\left(\tilde\x_i \tilde\x_i^\top\right)=1.$ It follows that Tyler's M-estimator $\hat \M$ with $\tr(\hat \M)=p$ is the solution to the following fixed-point equation: 
\beq \label{eq:Tyler_normalized_data}
\hat \M = \frac{1}{n}\sum_{i=1}^n \tilde w_i \tilde  \x_i \tilde \x_i^\top \ \mbox{ with } \
\tilde{w}_{i} 
=  \frac{ p \, [\tilde \x_{i}^\top \hat \M^{-1} \tilde \x_{i}]^{-1} } { \frac{1}{n}\sum_{\ell=1}^n  [\tilde \x_\ell^\top \hat \M^{-1} \tilde \x_\ell]^{-1} },
\eeq
where $\tilde{w}_{\ell}$'s are the normalized Tyler's weights computed on normalized (unit norm) observations $\tilde \x_i$'s. Thus, while Tyler's M-estimator $\hat \M$ can be interpreted as a weighted SCM based on normalized data $\{\tilde \x_i\}$ as shown in \eqref{eq:Tyler_normalized_data},  TWE of scatter $\hat \M_{\textup{TWE}}$ in \eqref{eq:TylerScatter} can be viewed as weighted SCM of actual (non-normalized) data $\{\x_i\}$. 

It is worthwhile to point out that most robust estimators of scatter  are affine equivariant in the sense that an affine transformation on the data $\mathbf{x}_i  \mapsto \mathbf{A} \mathbf{x}_i$, $i = 1,\ldots,n$, induces following transformation on the estimate: 
\beq \label{eq:affeq_scatter}
\hat \M(\{ \A \x_i\})= \A \hat \M(\{\x_i\}) \A^\top, \ \forall \mathbf{A} \in \mathbb{R}^{p \times p} \mbox{ invertible.} 
\eeq
For example, robust Maronna's \cite{maronna1976robust} M-estimators, S-estimators \cite{davies1987asymptotic}, or MM-estimators \cite{tatsuoka2000uniqueness} are affine equivariant. 
However, Tyler's scatter matrix $\hat \M$ is not affine equivariant since   \eqref{eq:affeq_scatter} only holds up to multiplicative scalar factor as shown in \eqref{eq:Tyler_scatter_star}.  
Affine equivariance is desirable since if $\x \sim \mathcal E_p(\mathbf{0},\M, g) $, then $\A \x \sim \mathcal E_p(\mathbf{0},\A \M \A^\top, g)$. 
Hence the transformed data shares the same elliptical distribution, but the scatter matrix parameter is mapped to $\A \M \A^\top$. 
Thus a natural requirement to be imposed on any scatter matrix estimator is that it should verify this same equivariance  principle under transformations $\x_i \mapsto \A \x_i$. 
This is shown next.

\begin{lemma} TWE of scatter matrix $\hat \M_{\textup{TWE}}$ is affine equivariant:  $\hat \M_{\textup{TWE}}(\{ \A \x_i \})= \A \hat \M(\{\x_i\}) \A^\top$,  $\forall$ invertible $\mathbf{A} \in \mathbb{R}^{p \times p}$.
\end{lemma} 

\begin{proof} It is straightforward to verify that Tyler's M-estimator (with $\tr(\hat \M)=p$) is equivariant in the sense that if $\hat \M^* = \hat{\M}(\{ \x_i^*\})$ denotes Tyler's M-estimator (verifying $\tr(\hat \M^*) =p$)  computed on data $\x_i^*=\A \x_i$, $i=1,\ldots,n$, then 
\beq \label{eq:Tyler_scatter_star}
\hat \M^* =c  \A \hat \M \A^\top,  \ c=\frac{ p}{ \tr(\A \hat \M \A^\top)}
\eeq 
for all invertible $\A \in \mathbb{R}^{p\times p}$. 
Now let $\hat w_i^*$, $i=1,\ldots,n$   denote the corresponding Tyler's weights. 
Then the scale estimate \eqref{eq:hat_tau3} computed as the harmonic mean of  reciprocals of weights $1/\hat w_i^* = [ \x_i^*]^\top(\hat{\M}^*)^{-1} \x_i^*$ is 
\begin{align}
\hat \eta^*_{\textup{TWE}} &= \Big(\frac 1 n \sum_{i=1}^n \hat w_i^* \Big)^{-1} 
  =  \bigg( \frac{p}{n} \sum_{i=1}^n  [\x_{i}^\top \A^\top (\hat{\M}^*)^{-1} \A \x_{i}]^{-1} \bigg)^{-1}  \notag \\
 &= c^{-1} \bigg( \frac{p}{n} \sum_{i=1}^n  [\x_{i}^\top \hat{\M}^{-1}  \x_{i}]^{-1} \bigg)^{-1}  = c^{-1} \hat{\eta}_{\textup{TWE}} 
\end{align}
where in the 2nd to last identity we simply utilized \eqref{eq:Tyler_scatter_star}. Thus it follows that 
\begin{align*}
\hat{\M}_{\textup{TWE}}^* &= \hat \eta_{\textup{TWE}}^* \cdot \hat \M^* = c^{-1} \hat \eta_{\textup{TWE}} \cdot  c  \A \hat{\M}\A^{\top} \\ &= \A ( \hat \eta_{\textup{TWE}}\hat \M)\A^\top 
= \A \hat\M_{\textup{TWE}} \A^{\top}.
\end{align*}
\end{proof} 

In the case of large dimensional data, one can also consider a shrinkage Tyler's M-estimator of the scatter matrix as  
\beq \label{eq:shrink_tyl}
\hat \M_{\textup{TWE},\beta} = \beta \hat \M_{\textup{TWE}}  + (1-\beta) \hat \eta_{\textup{TWE}} \bo I, 
\eeq 
where the data adaptive shrinkage parameter $\beta \in [0,1]$ is computed as described in \cite[Sect.~IV.C]{ollila2021shrinking}. However, unlike the estimator 
in \cite{ollila2021shrinking}, the shrinkage TWE in  \eqref{eq:shrink_tyl} provides  an estimator  of scatter instead of shape matrix.

\subsection{An estimator of covariance matrix} 

If the  density generator $g$ (and hence the underlying ES distribution) is specified, then the value of  $\theta$ in \eqref{eq:theta} can be determined,  and  we can use relationship \eqref{eq:theta} to obtain a covariance matrix estimator as 
$
\hat{\mathbf{R}}_{\text{TWE}}  = \theta \cdot \hat \M_{\text{TWE}} . 
$
For example, if the data has an MVN distribution,  then 
$\theta=1$ while $\theta = \nu/(\nu-2)$ in the case of an MVT distribution with $\nu$ d.o.f.  
However, often the underlying parametric family is known, but the underlying tail parameter, say $\nu$, indexing the density generator is unknown. 
As is shown in \autoref{sec:tail}, we can form an estimate of $\nu$, denoted $\hat \nu_{\text{TWE}}$, using Tyler's weights. 
Since $\theta = h(\nu)$ ({\it cf.} Eq. \eqref{eq:eta_ratio} below), a \emph{TWE of covariance matrix} can be computed as 
$\hat{\mathbf{R}}_{\textup{TWE}}  =  \hat \theta_{\textup{TWE}} \cdot \hat \M_{\textup{TWE}}$,   
where $ \hat \theta_{\text{TWE}} = h( \hat \nu_{\text{TWE}})$.

\section{Estimating the tail parameter of ES distribution} \label{sec:tail}

From \eqref{eq:theta} we can induce the following relationship between the scale parameter $\eta _{\text{cov}} = p^{-1}\tr(\mathbf{R})$ of the covariance matrix and scale $\eta = p^{-1}\tr(\M)$ of the scatter matrix:
\beq \label{eq:eta_rat}
\eta_{\text{cov}} =  \theta  \eta    \Leftrightarrow \theta = \eta_{\text{cov}}/\eta .
\eeq 
Note that  a natural estimate of $\eta_{\text{cov}}$ is $p^{-1}\tr(\mathbf{S})$, where  $\mathbf{S} = \frac 1 n \sum_{i=1}^n \x_i \x_i^\top$ denotes the sample covariance matrix (SCM). On the other hand,  if density generator $g$ is specified up to unknown tail parameter $\nu$, thus indexed by $g_{\nu}(\cdot)$, then $\theta$ in \eqref{eq:theta} is a following function of the tail parameter $\nu$:
\beq \label{eq:eta_ratio}
\begin{aligned}
\theta  = \int_0^{\infty} t f_{r^2}(t; \nu) \mathrm{d} t 
\triangleq h(\nu)
\end{aligned}
\eeq
where the pdf $f_{r^2}(\cdot)$, defined in \eqref{eq:rsq_pdf}, is one-to-one with $g_\nu(\cdot)$. We do not need numerical integration in most practical cases as often closed-form expression for $h(\nu)$ can be derived.  Then, after solving the inverse mapping, $\nu = h^{-1}(\theta) = h^{-1}(\eta_{\text{cov}}/\eta)$, Algorithm \ref{alg:tail_param} offers a unified approach for estimating the tail parameter of an ES distribution:

\begin{algorithm}[!h]
\caption{Distribution tail parameter learning}\label{alg:tail_param}
\textbf{Input:} Data  $\{\x_i\}_{i=1}^n$ \\
\textbf{Output:} Estimated tail parameter $\nu$
\begin{enumerate}
\item[1.] Compute Tyler's M-estimator  $\hat \M$ and weights $\hat w_i$'s in \eqref{eq:Mest-tyler};
\item[2.] Compute $\hat \eta_{\text{TWE}}$ in \eqref{eq:hat_tau3}, and set  $\hat \theta_{\text{TWE}} = p^{-1}\tr(\S)/\hat \eta_{\text{TWE}}$;
\item[3.] Using \eqref{eq:eta_ratio}, estimate $\nu$ as $ \hat \nu_{\text{TWE}} = h^{-1}(\hat \theta_{\text{TWE}})$.
\end{enumerate}
\end{algorithm}

As an example, if $\x$ follows an MVT distribution with $\nu>2$ d.o.f., one has that  $\theta = h(\nu)=  \frac{\nu}{\nu -2}$, which 
unfolds the relation: 
\beq \label{eq:eta_ratio2}
\nu = h^{-1}(\theta) = \frac{2 \theta}{ \theta-1} \mbox { for }  \theta >1.
\eeq 
Note that $\nu >2$ is required for the covariance matrix $\mathbf{R}$ to exist. The obtained estimator $ \hat \nu_{\text{TWE}}$  is closely related to estimator in \cite[Alg.~1]{ollila2021shrinking}, referred to as \textbf{OPP estimator} for short.  OPP is an iterative approach that iteratively (re-)computes the maximum likelihood estimator (MLE)  $\hat \M$ of the MVT distribution with $\nu$ given by the current estimate of d.o.f. parameter $\nu^{(k)}$. It then computes $ \hat \theta = \tr(\S)/\tr(\hat \M)$ which provides an update $\nu^{(k+1)}=h^{-1}(\hat \theta)$ via \eqref{eq:eta_ratio2}. The algorithm iterates for $k=0,1,2 \ldots$ until convergence, 
starting from an initial start $\nu^{(0)} = \hat \nu_{\text{kurt}}$, where  $\hat \nu_{\text{kurt}}$  is an estimate of $\nu$ based on elliptical kurtosis,  proposed in  \cite{ollila2021shrinking}, and referred to as \textbf{kurtosis estimator}. We also proposed an improved version of OPP estimator in \cite{pascal-EUSIPCO2021}, which, however, is impractical for large $n$ and $p$. ML estimation of $\nu$ via the Expectation-Maximization (EM) approach  is considered in  \cite{LiuRubin95}. This method is unfortunately rather unstable \cite{FernandezSteel99}.
 
\section{Simulation studies} \label{sec:simu}

\begin{figure}[!t]
\setlength\fwidth{1.0\columnwidth}
\centerline{\input{tikz/MVTnu_v=5_p=100_n=600_MC=5000_rho=0dot6.tex}} \vspace{-0.1cm} 
\centerline{\input{tikz/MVTnu_v=3_p=100_n=600_MC=5000_rho=0dot6.tex}} \vspace{-0.1cm} 
  \caption{Boxplots d.o.f. estimates $\hat \nu$ as a function of $n$ when $\nu=5$ (top) and $\nu=3$ (bottom);  $p=100$, $\rho=0.6$.}
  \label{fig:boxplot_nu}
  \setlength\fwidth{1.0\columnwidth}
\centerline{\begin{tikzpicture}

\definecolor{darkgray176}{RGB}{176,176,176}
\definecolor{darkorange25512714}{RGB}{255,127,14}
\definecolor{forestgreen4416044}{RGB}{44,160,44}
\definecolor{lightgray204}{RGB}{204,204,204}
\definecolor{steelblue31119180}{RGB}{31,119,180}
 \pgfmathsetlengthmacro\MajorTickLength{
      \pgfkeysvalueof{/pgfplots/major tick length} * 0.5
    }

\begin{axis}[
width=0.56\fwidth,
height=4.7cm, 
 xlabel shift={-3pt},
label style={font=\footnotesize} , 
tick label style={font=\scriptsize} , 
legend cell align={left},
yticklabel shift=-2pt,
legend style={
  fill opacity=0.7,
  draw opacity=1,
  text opacity=1,
  font= \footnotesize,
  at={(0.5,0.5)},
  anchor=center,
  draw=lightgray204
},
tick align=outside,
tick pos=left,
every tick/.style={
        black,
        semithick,
      },
      major tick length=\MajorTickLength,
x grid style={darkgray176},
xlabel={\(\displaystyle n\)},
xmin=127.5, xmax=622.5,
xtick style={color=black},
y grid style={darkgray176},
ylabel={mean squared error of \(\displaystyle \hat \nu\)},
ymin=0, ymax=1.8,
ytick={0,0.6,1.2,1.8},
%yticklabels={,,},
%xmajorgrids,
%ymajorgrids,
ytick style={color=black}
]

\addplot [semithick, forestgreen4416044, dashed, mark=*, mark size=2, mark options={solid}]
table {%
150 3.17536140598538
200 2.38279758633393
250 1.99045064683226
300 1.72678114257841
600 0.936936262978288
};
\addlegendentry{kurtosis}
\addplot [semithick, darkorange25512714, dashed, mark=*, mark size=2, mark options={solid}]
table {%
150 1.24990958250195
200 0.645913602379811
250 0.474801410488116
300 0.376468458255817
600 0.168580882748486
};
\addlegendentry{OPP}
\addplot [semithick, steelblue31119180, dashed, mark=*, mark size=2, mark options={solid}]
table {%
150 0.39930567236767
200 0.289894829287337
250 0.241894787612174
300 0.202821659791771
600 0.117010066052748
};
\addlegendentry{proposed}
\end{axis}

\end{tikzpicture}\begin{tikzpicture}

\definecolor{darkgray176}{RGB}{176,176,176}
\definecolor{darkorange25512714}{RGB}{255,127,14}
\definecolor{forestgreen4416044}{RGB}{44,160,44}
\definecolor{lightgray204}{RGB}{204,204,204}
\definecolor{steelblue31119180}{RGB}{31,119,180}
\pgfmathsetlengthmacro\MajorTickLength{
      \pgfkeysvalueof{/pgfplots/major tick length} * 0.5
    }

\begin{axis}[
width=0.6\fwidth,
height=4.7cm, 
xlabel shift={-3pt},
label style={font=\scriptsize} , 
tick label style={font=\scriptsize} , 
legend cell align={left},
legend style={fill opacity=0.8, draw opacity=1, text opacity=1,font=    \footnotesize, draw=lightgray204},
tick align=outside,
tick pos=left,
every tick/.style={
        black,
        semithick,
      },
      major tick length=\MajorTickLength,
yticklabel shift=-2pt,
x grid style={darkgray176},
xlabel={\(\displaystyle n\)},
xmin=127.5, xmax=622.5,
xtick style={color=black},
y grid style={darkgray176},
ymin=0, ymax=0.4,
ytick style={color=black}
]

\addplot [semithick, darkorange25512714, dashed, mark=*, mark size=2, mark options={solid}]
table {%
150 0.380414598266082
200 0.243638437569662
250 0.187818882696393
300 0.157916803920662
600 0.0885782986127917
};
\addlegendentry{OPP}
\addplot [semithick, steelblue31119180, dashed, mark=*, mark size=2, mark options={solid}]
table {%
150 0.18386453283228
200 0.141060725019736
250 0.115666783511601
300 0.10053145252902
600 0.0600563459378289
};
\addlegendentry{proposed}
\end{axis}

\end{tikzpicture}} \vspace{-0.1cm}
  \caption{MSE of d.o.f estimates $\hat \nu$ as a function of $n$ when $\nu=5$ (left panel) and $\nu=3$ (right panel);  $p=100$, $\rho=0.6$.}
  \label{fig:MSEnu}
\end{figure}
We first consider the case where scatter matrix 
$\M$ has an autoregressive model (\textbf{AR(1)}) structure,
$(\boldsymbol{\Sigma})_{ij} = \eta \varrho^{|i-j|}$,
where $\eta = \tr(\M)/p$ is the scale parameter  and $\varrho$ is the correlation parameter, $\varrho \in (-1,1)$.   Since Tyler's M-estimator $\hat \M$ is invariant to the data scaling, we can set $\eta=1$ without favoring any estimator over the other.  The number of  Monte-Carlo runs is 5,000, and samples are generated from an MVT distribution with different choices of d.o.f. parameter $\nu$.

First, we investigate how the TWE of d.o.f. parameter $\nu$ compares against OPP and kurtosis estimators. \autoref{fig:boxplot_nu} displays the boxplots in the case that $p=100$ and varying sample lengths when  
$\nu=5$ or $\nu=3$. 
As can be noted, the proposed TWE attains the best accuracy as well as the smallest variability. Moreover, for $\nu=5$, its median values are right on  the spot. The kurtosis estimator obviously performs poorly when $\nu=3$ since the 4th-order moment does not exist in this case. \autoref{fig:MSEnu} shows the average mean squared error (MSE),  $(\hat  \nu - \nu)^2$, which further illustrates the benefits and high accuracy of the proposed TWE against its competitors.

\begin{figure}
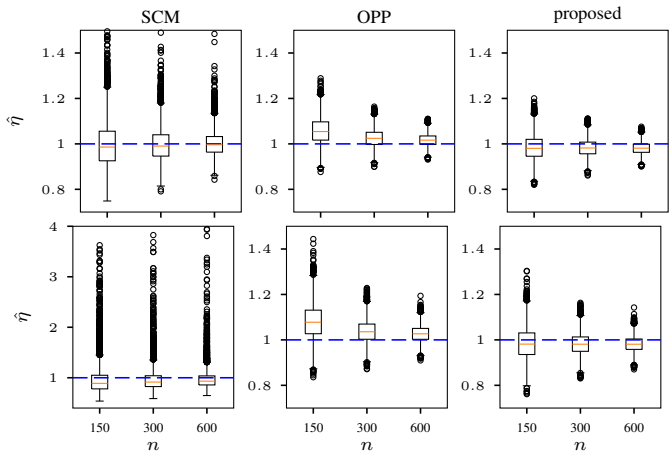

  \setlength\fwidth{1.0\columnwidth}
\centerline{\input{tikz/MVTtau_v=5_p=100_n=600_MC=5000_rho=0dot6.tex}}  \vspace{-0.3cm}
\centerline{\input{tikz/MVTtau_v=3_p=100_n=600_MC=5000_rho=0dot6.tex}} \vspace{-0.2cm}
  \caption{Estimation errors  of scale estimates $\hat \eta$ as a function of $n$ when $\nu=5$ (top) and $\nu=3$ (bottom);  $p=100$, $\rho=0.6$.}
  \label{fig:MSE_v5}
\end{figure}

 \autoref{fig:MSE_v5} displays the boxplots of different estimates of scale $\eta$. Here we compare $\hat \eta_{\text{TWE}}$ to OPP estimate of scale, defined as  
$\hat \eta_{\text{OPP}} = \tr(\hat \M)/p $, where $\hat \M$ is the MLE of scatter based on $\nu = \hat \nu_{\text{OPP}}$.   We also compare with the  scale estimate provided by the SCM, defined as $\hat \eta_{\cov} = \tr(\S)/p$, but multiplied by $\theta^{-1}=(\nu-2)/\nu$ to obtain an estimate of $\eta$; recall \eqref{eq:eta_rat}. 
We can notice from \autoref{fig:MSE_v5}  that TWE slightly underestimates the true scale while  OPP is overestimating. We also notice that the SCM estimator is clearly unbiased for $\nu=5$, but has huge variability. \autoref{fig:boxplot_nu_pic2} displays the median values of $\hat \nu$ 
 for range of $\nu$ values when $n=150$ (and $p=100$ and $\varrho=0.6$ as earlier). The proposed TWE estimator 
  significantly outperforms the other estimators for all d.o.f. $\nu \in [3,8]$.

\begin{figure}
\centerline{\begin{tikzpicture}

\definecolor{darkgray176}{RGB}{176,176,176}
\definecolor{darkorange25512714}{RGB}{255,127,14}
\definecolor{forestgreen4416044}{RGB}{44,160,44}
\definecolor{lightgray204}{RGB}{204,204,204}
\definecolor{steelblue31119180}{RGB}{31,119,180}

\begin{axis}[
width=0.9\columnwidth,
height=5.7cm, 
legend cell align={left},
legend style={
  fill opacity=0.8,
  draw opacity=1,
  text opacity=1,
  at={(0.03,0.97)},
  anchor=north west,
  draw=lightgray204
},
tick align=outside,
tick pos=left,
x grid style={darkgray176},
xlabel={\(\displaystyle \nu\)},
label style={font=\footnotesize} , 
tick label style={font=\footnotesize} , 
xmajorgrids,
xmin=2.75, 
xmax=8.25,
xtick style={color=black},
y grid style={darkgray176},
ylabel={\(\displaystyle \hat \nu\)},
ymajorgrids,
ymin=2.36944785925033, ymax=11.3962385490334,
ytick style={color=black}
]
\path [draw=forestgreen4416044, fill=forestgreen4416044, opacity=0.2]
(axis cs:3,5.3002816777344)
--(axis cs:3,4.32825770158366)
--(axis cs:4,4.80677889310405)
--(axis cs:5,5.47649808248211)
--(axis cs:6,6.1335379524786)
--(axis cs:7,6.92167483915153)
--(axis cs:8,7.61396469840647)
--(axis cs:8,10.985929881316)
--(axis cs:8,10.985929881316)
--(axis cs:7,9.81182522301371)
--(axis cs:6,8.5858476425845)
--(axis cs:5,7.43685375463398)
--(axis cs:4,6.29771857660215)
--(axis cs:3,5.3002816777344)
--cycle;

\path [draw=darkorange25512714, fill=darkorange25512714, opacity=0.2]
(axis cs:3,3.85097685674525)
--(axis cs:3,2.87858312053963)
--(axis cs:4,3.80277742328872)
--(axis cs:5,4.81109274390899)
--(axis cs:6,5.78278972576549)
--(axis cs:7,6.83745542376691)
--(axis cs:8,7.86756710696241)
--(axis cs:8,10.9671715082214)
--(axis cs:8,10.9671715082214)
--(axis cs:7,9.46792178608334)
--(axis cs:6,7.93889460490887)
--(axis cs:5,6.51388754280015)
--(axis cs:4,5.13611158845396)
--(axis cs:3,3.85097685674525)
--cycle;

\path [draw=steelblue31119180, fill=steelblue31119180, opacity=0.2]
(axis cs:3,3.56097567435208)
--(axis cs:3,2.77975652696774)
--(axis cs:4,3.56952214602384)
--(axis cs:5,4.39146107150736)
--(axis cs:6,5.17492623258474)
--(axis cs:7,5.97610719031463)
--(axis cs:8,6.72125550728119)
--(axis cs:8,8.62256217610409)
--(axis cs:8,8.62256217610409)
--(axis cs:7,7.67075530646851)
--(axis cs:6,6.63813430796652)
--(axis cs:5,5.61123092124765)
--(axis cs:4,4.58147512004147)
--(axis cs:3,3.56097567435208)
--cycle;

\addplot [semithick, black, forget plot]
table {%
3 3
4 4
5 5
6 6
7 7
8 8
};
\addplot [semithick, forestgreen4416044, mark=*, mark size=1.7, mark options={solid}]
table {%
3 4.81426968965903
4 5.5522487348531
5 6.45667591855805
6 7.35969279753155
7 8.36675003108262
8 9.29994728986122
};
\addlegendentry{{\small kurtosis}}
\addplot [semithick, darkorange25512714, mark=*, mark size=1.7, mark options={solid}]
table {%
3 3.36477998864244
4 4.46944450587134
5 5.66249014335457
6 6.86084216533718
7 8.15268860492512
8 9.41736930759191
};
\addlegendentry{{\small OPP}}
\addplot [semithick, steelblue31119180, mark=*, mark size=1.7, mark options={solid}]
table {%
3 3.17036610065991
4 4.07549863303266
5 5.0013459963775
6 5.90653027027563
7 6.82343124839157
8 7.67190884169264
};
\addlegendentry{{\small proposed}}
\end{axis}

\end{tikzpicture}}
  \vspace{-0.1cm}
  \caption{Estimated values of $\nu$; $p=100$, $\rho=0.6$, $n=150$.  The black line indicates the true value.}
  \label{fig:boxplot_nu_pic2}
  \vspace{0.4cm}
   \centerline{\begin{tikzpicture}

\definecolor{darkgray176}{RGB}{176,176,176}
\definecolor{darkorange25512714}{RGB}{255,127,14}
\definecolor{forestgreen4416044}{RGB}{44,160,44}
\definecolor{steelblue31119180}{RGB}{31,119,180}
\definecolor{lightgray204}{RGB}{204,204,204}

\begin{axis}[
width=0.9\columnwidth,
height=5.7cm, 
legend cell align={left},
legend style={
  fill opacity=0.8,
  draw opacity=1,
  text opacity=1,
  draw=lightgray204
},
tick align=outside,
tick pos=left,
x grid style={darkgray176},
xlabel={\(\displaystyle n\)},
xmajorgrids,
xmin=137.5, xmax=412.5,
xtick style={color=black},
y grid style={darkgray176},
label style={font=\footnotesize} , 
tick label style={font=\footnotesize} , 
ylabel={\(\displaystyle \hat \nu\)},
ymajorgrids,
ymin=3.8, ymax=6.2,
ytick style={color=black}
]
\path [draw=forestgreen4416044, fill=forestgreen4416044, opacity=0.2]
(axis cs:150,5.83311797839665)
--(axis cs:150,5.40131541840231)
--(axis cs:200,5.28077168654434)
--(axis cs:250,5.1708993243812)
--(axis cs:300,5.12440660645236)
--(axis cs:350,5.06710037829106)
--(axis cs:400,5.03626802337571)
--(axis cs:400,5.34700881536809)
--(axis cs:400,5.34700881536809)
--(axis cs:350,5.38886219181487)
--(axis cs:300,5.45968112278172)
--(axis cs:250,5.52570943551684)
--(axis cs:200,5.6706837578558)
--(axis cs:150,5.83311797839665)
--cycle;

\path [draw=darkorange25512714, fill=darkorange25512714, opacity=0.2]
(axis cs:150,4.45794966848282)
--(axis cs:150,4.12251622398491)
--(axis cs:200,4.0723028371725)
--(axis cs:250,4.02905690762134)
--(axis cs:300,4.03721860862135)
--(axis cs:350,4.01784954364334)
--(axis cs:400,4.0183305733582)
--(axis cs:400,4.20488535812984)
--(axis cs:400,4.20488535812984)
--(axis cs:350,4.22062055424716)
--(axis cs:300,4.24416716558802)
--(axis cs:250,4.27862448628023)
--(axis cs:200,4.33858474825184)
--(axis cs:150,4.45794966848282)
--cycle;

\path [draw=steelblue31119180, fill=steelblue31119180, opacity=0.2]
(axis cs:150,4.1898968918087)
--(axis cs:150,3.91420614747133)
--(axis cs:200,3.93311249924626)
--(axis cs:250,3.90986532320852)
--(axis cs:300,3.92951871443009)
--(axis cs:350,3.9177897244814)
--(axis cs:400,3.91971771703731)
--(axis cs:400,4.09229038756077)
--(axis cs:400,4.09229038756077)
--(axis cs:350,4.10356166582392)
--(axis cs:300,4.1160177748655)
--(axis cs:250,4.1354587390716)
--(axis cs:200,4.16505951429453)
--(axis cs:150,4.1898968918087)
--cycle;

\addplot [semithick, black,forget plot]
table {%
150 4
200 4
250 4
300 4
350 4
400 4
};
\addplot [semithick, forestgreen4416044, mark=*, mark size=1.7, mark options={solid}]
table {%
150 5.60354667918287
200 5.46748556992767
250 5.33150236114619
300 5.27830772381588
350 5.23290850048653
400 5.18516035192373
};
\addlegendentry{{\small kurtosis}}
\addplot [semithick, darkorange25512714, mark=*, mark size=1.7, mark options={solid}]
table {%
150 4.29242159202285
200 4.20115035088374
250 4.15095420101023
300 4.13919874659229
350 4.12071419805098
400 4.11435254946831
};
\addlegendentry{{\small OPP}}
\addplot [semithick, steelblue31119180, mark=*, mark size=1.7, mark options={solid}]
table {%
150 4.05743367894343
200 4.04553184904362
250 4.02177497160542
300 4.02197333958385
350 4.01224381857264
400 4.00702288592852
};
\addlegendentry{{\small proposed}}
\end{axis}

\end{tikzpicture}}
  \caption{Estimated values of $\nu$ for synthetic stock market data; $p=100, \nu=4$.}
  \label{fig:SP500-nu-vs-n}
\end{figure}

We now consider an example based on stock market data. We generate synthetic data ($p=100$ assets) with heavy tails following MVT distribution with d.o.f. $\nu=4$ and covariance matrix as measured from stocks of S\&P 500 index. Figure \ref{fig:SP500-nu-vs-n} compares the estimated value of $\nu$ versus the number of observations for the following methods: kurtosis estimator, OPP estimator \cite{ollila2021shrinking}, and the proposed estimator, with the latter being clearly superior, illustrating its promising performance for real-world financial data.

\section{Concluding remarks}   \label{sec:conclusion}

We  proposed a new robust estimator of scale parameter of an elliptical distribution based on the weights from Tyler's M-estimator, which was further used to construct an affine equivariant Tyler's M-estimator.
We then proposed a unified framework to estimate the tail parameter of an elliptical distribution. 
Finally, it should be noted that this method generalizes to complex-valued data in a straightforward manner.

\end{document}